\documentclass[11pt,a4paper]{article}

\usepackage{amsmath,amssymb,amsthm}
\usepackage{hyperref}
\usepackage{enumitem}
\usepackage[numbers,sort&compress]{natbib}
\usepackage{geometry}
\usepackage{microtype}
\usepackage{titlesec}
\usepackage{abstract}

\titleformat{\section}{\normalfont\large\bfseries}{\thesection.}{0.5em}{}
\titleformat{\subsection}{\normalfont\normalsize\bfseries}{\thesubsection.}{0.5em}{}

\numberwithin{equation}{section}

\newtheorem{Theorem}{Theorem}[section]
\newtheorem*{Theorem*}{Theorem}
\newtheorem{Corollary}[Theorem]{Corollary}

\newtheorem{Proposition}[Theorem]{Proposition}

\theoremstyle{definition}
\newtheorem{Definition}[Theorem]{Definition}

\newtheorem{Remark}[Theorem]{Remark}

\newcommand{\cA}{\mathcal{A}}
\newcommand{\cH}{\mathcal{H}}
\newcommand{\cD}{\mathcal{D}}
\newcommand{\cL}{\mathcal{L}}
\newcommand{\Qg}{\mathcal{Q}_\gamma}
\newcommand{\Dg}{\mathcal{D}_\gamma}
\newcommand{\Sig}{\Sigma}
\newcommand{\Tr}{\operatorname{Tr}}
\newcommand{\tr}{\operatorname{tr}}
\newcommand{\Dom}{\operatorname{Dom}}

\begin{document}

\title{\textbf{Lindblad-Deformed Spectral Geometry: Heat-Kernel
    Asymptotics and Effective Spectral Dimension}}

\author{Soumadeep Maiti\thanks{E-mail:
  \href{mailto:maitisoumadeep@gmail.com}{\texttt{maitisoumadeep@gmail.com}}.}\\[4pt]
  \small\itshape Ludwig Maximilian University
  of Munich, Geschwister-Scholl-Platz 1 80539 München, Germany}

\date{}
\maketitle
\thispagestyle{empty}

\begin{abstract}
\noindent
We introduce a Lindblad-deformed spectral geometric framework in which
bounded dissipative data deform a standard spectral triple through the
Dirac operator $\Dg=\cD-i\gamma\Sig$, where
$\Sig=\frac{1}{2}\sum_k L_k^\dagger L_k$ is constructed from Lindblad
jump operators $\{L_k\}$. The associated positive operator
$\Qg=\Dg^*\Dg=\cD^2+\gamma^2\Sig^2-i\gamma[\cD,\Sig]$ is identified as
the correct spectral-geometric observable. For smooth endomorphism-valued
Lindblad data, $\Qg$ is of Laplace type and therefore admits a standard
heat-kernel asymptotic expansion with dissipation-modified even
Seeley--DeWitt coefficients. For the scalar deformation
$L=\sqrt{\gamma}\,f$ with $f\in C^\infty(M)$ real-valued, we prove that
the first-order Duhamel correction to the heat trace
$K_\gamma(\sigma)=\Tr(e^{-\sigma\Qg})$ vanishes identically, so that the
first nontrivial dissipative effect appears at order $\gamma^4$. We then
identify the exact Duhamel-level decomposition of the $O(\gamma^4)$
correction into a direct $W_2$ insertion and a quadratic $W_1\otimes W_1$
term. In the round $S^2$ model we determine the explicit deformed operator
and extract the leading local asymptotic contribution of the $W_2$ sector,
while the quadratic sector contributes at leading order as $O(\sigma)$ in
two dimensions. Finally, we define the effective scale-dependent spectral
dimension $d_{s,\mathrm{eff}}(\sigma,\gamma)=
-2\,\partial_{\log\sigma}\log K_\gamma(\sigma)$ and identify its leading
perturbative deformation. The framework provides a mathematically
controlled starting point for incorporating open-system effects into
spectral geometry.


\end{abstract}

\section{Introduction}
\label{sec:intro}

Spectral geometry packages metric information into the analytic behaviour
of Dirac- and Laplace-type operators. In the noncommutative setting this
principle is encoded by spectral triples, whose algebraic and analytic
data recover large parts of Riemannian geometry and support the spectral
action programme~\cite{Connes1994,Connes1995,ConnesMarcolli2008,
ChamseddineConnes1997}. At the same time, open quantum systems are most
naturally described by quantum dynamical semigroups generated by
Gorini--Kossakowski--Sudarshan--Lindblad (Lindblad)
operators~\cite{Lindblad1976,GoriniKossakowskiSudarshan1976}. The present
paper addresses a deliberately modest but precise question: if one
deforms the Dirac operator by bounded Lindblad dissipation data, what is
the correct spectral-geometric object, and which heat-kernel structures
survive this deformation?

A first obstacle appears immediately. If $\cL$ is a Lindblad generator
acting on density matrices and $\rho_0$ is a normalised initial state,
then
\begin{equation}
\Tr(e^{t \cL}\rho_0)=\Tr(\rho_0)=1, \qquad t\ge 0,
\end{equation}
by trace preservation. This quantity is therefore spectrally trivial and
cannot serve as a useful geometric probe. The correct spectral-geometric
object is instead the positive operator
\begin{equation}
\Qg:=\Dg^*\Dg,
\qquad
\Dg:=\cD-i\gamma\Sig,
\end{equation}
constructed on the original spinor Hilbert space. The resulting heat trace
\begin{equation}
K_\gamma(\sigma):=\Tr\!\left(e^{-\sigma \Qg}\right)
\end{equation}
is the central analytic observable studied in this paper.

Our work lies near, but not within, several existing generalisations of
standard spectral triples. Symmetric non-self-adjoint and pre-spectral
variants have been developed in operator-algebraic directions
\cite{ConnesLevitinaEtAl2019,LiWang2025}; modular and thermodynamic
enrichments of spectral geometry have also been explored
\cite{BertozziniContiLewkeeratiyutkul2010}; and semigroup-based
constructions related to heat kernels and almost-commutative geometry
have appeared in recent work~\cite{Gakkhar2022,Gakkhar2024}. On the
physics side, spectral dimension has long been used as a probe of scale
dependence in quantum-gravity-inspired
settings~\cite{AmbjornJurkiewiczLoll2005,AlkoferSaueressigZanusso2015}.
The novelty of the present paper is more specific. We do not propose a
full new axiomatics of open spectral triples, nor do we claim primacy
among frameworks beyond self-adjoint spectral triples. Rather, we
identify $\Qg=\Dg^*\Dg$ as the correct positive spectral-geometric
observable associated with a bounded Lindblad deformation of the Dirac
operator, and in the smooth endomorphism-valued setting prove that it is of
Laplace type and supports standard heat-kernel asymptotics. To the best
of our knowledge, the present work is the first to identify
$\Qg=\Dg^*\Dg$ as the central positive spectral-geometric observable
associated with a bounded Lindblad deformation of the Dirac operator,
and to analyze its heat-kernel asymptotics, scalar perturbative rigidity
at order $\gamma^2$, and leading $O(\gamma^4)$ structure in this
framework.

The key technical result is that in the scalar case the order-$\gamma^2$
correction to the heat trace vanishes identically, so that the first
nontrivial dissipative effect appears only at order $\gamma^4$.

The main results are as follows. Under explicit boundedness assumptions
on the Lindblad data, Proposition~\ref{prop:Qgamma} gives the exact
identity
\begin{equation}
\Qg=\cD^2+\gamma^2\Sig^2-i\gamma[\cD,\Sig].
\end{equation}
For scalar deformations $L=\sqrt\gamma\,f$,
Proposition~\ref{prop:vanishing} shows that the first Duhamel correction
to $K_\gamma$ vanishes identically, so that the first nontrivial
perturbative contribution appears at order $\gamma^4$.
Corollary~\ref{cor:SDW} upgrades this to a coefficient-level statement:
$A_0(\Qg)=A_0(\cD^2)$ exactly, and $A_2(\Qg)=A_2(\cD^2)+O(\gamma^4)$,
so there is no order-$\gamma^2$ correction to the first nontrivial
Seeley--DeWitt coefficient. In the round $S^2$
toy model with $f(\theta,\phi)=\cos\theta$, we determine the explicit
deformed operator, isolate the two distinct $O(\gamma^4)$ sectors at
the Duhamel level, and identify at the level of leading local structure
the scalar integrals $\int_M f^4\,d\mathrm{vol}_g$ and
$\int_M f^2|\nabla f|^2\,d\mathrm{vol}_g$ as the natural geometric
quantities arising from the two sectors (Remark~\ref{rem:localinvariants}).
Finally, we define the effective scale-dependent
spectral dimension associated with $K_\gamma$ and identify its leading
perturbative deformation.

The paper is organised as follows. Section~\ref{sec:background} reviews
the spectral-triple and heat-kernel background. Section~\ref{sec:deformation}
introduces the Lindblad deformation and proves the basic
operator-theoretic properties of~$\Qg$. Section~\ref{sec:heatkernel}
shows that, in the smooth endomorphism-valued setting, $\Qg$ is of
Laplace type and supports a standard even heat-kernel expansion. Section~\ref{sec:vanishing} treats scalar Lindblad
data and proves the vanishing proposition. Section~\ref{sec:gamma4}
identifies the two $O(\gamma^4)$ sectors. Section~\ref{sec:S2}
specialises to the round $S^2$ and identifies the local form of the
deformed operator together with the leading small-$\sigma$ asymptotic
structure of the order-$\gamma^4$ correction. Section~\ref{sec:spectraldim}
defines the effective spectral dimension and extracts its perturbative
structure. Section~\ref{sec:discussion} summarises what has been
established and what remains open.

\section{Background}
\label{sec:background}

A spectral triple $(\cA,\cH,\cD)$ consists of a unital involutive
algebra $\cA$ represented on a separable Hilbert space $\cH$ together
with an unbounded self-adjoint operator $\cD$ whose resolvent is compact
modulo the algebra and whose commutators with algebra elements are
bounded~\cite{Connes1994,Connes1995,ConnesMarcolli2008}. In the
commutative model case, one has $\cA=C^\infty(M)$, $\cH=L^2(M,S)$, and
$\cD$ equal to the Dirac operator on a compact Riemannian spin manifold
$(M,g)$. The spectral data then encode the metric, measure, and dimension
in analytic form.

\begin{Definition}
A \emph{spectral triple} $(\cA,\cH,\cD)$ consists of a unital
$*$-algebra $\cA$ represented on a separable Hilbert space $\cH$ and a
self-adjoint operator $\cD \colon \Dom(\cD) \subset \cH \to \cH$
(the \emph{Dirac operator}) such that $a(\cD^2+1)^{-1/2} \in
\mathcal{K}(\cH)$ for all $a \in \cA$, and $[\cD, a] \in
\mathcal{B}(\cH)$ for all $a \in \cA$.
\end{Definition}

For the purposes of the present paper we remain entirely within the
compact spin-manifold setting. We therefore do not require the full
real-structure, orientability, or related auxiliary data of a general
spectral triple, although the conceptual background remains the same.

\subsection{Heat kernels and Seeley--DeWitt asymptotics}

If $P$ is a positive self-adjoint Laplace-type operator on a compact
manifold, then $e^{-\sigma P}$ is trace class for every $\sigma>0$, and
the heat trace admits the standard asymptotic expansion
\begin{equation}
\Tr(e^{-\sigma P})\sim (4\pi\sigma)^{-d/2}\sum_{m\ge 0}\sigma^m A_{2m}(P),
\qquad \sigma\to 0^+,
\end{equation}
where the coefficients $A_{2m}(P)$ are integrals of local curvature
invariants built from the metric and the lower-order part of
$P$~\cite{Vassilevich2003,Gilkey1995,BerlineGetzlerVergne1992}. On
closed manifolds only even powers occur in this standard Laplace-type
setting. In the present paper we remain entirely within that framework
and do not claim the appearance of odd coefficients.

\subsection{Existing generalisations}

Several neighbouring constructions should be kept in view. Pre-spectral
and symmetric non-self-adjoint frameworks replace the Dirac operator by
more general closed operators while preserving substantial parts of the
analytic spectral apparatus~\cite{ConnesLevitinaEtAl2019,LiWang2025}.
Modular and thermodynamic enrichments insert Kubo--Martin--Schwinger
(KMS) or modular data into the spectral-geometric
package~\cite{BertozziniContiLewkeeratiyutkul2010}. These developments
show that one can move beyond the strictly self-adjoint Dirac picture
without abandoning spectral methods, but they are not equivalent to the
bounded Lindblad deformation considered here.

\section{The Lindblad deformation and the operator
  $\mathcal{Q}_\gamma$}
\label{sec:deformation}

Let $(M^d,g)$ be a compact Riemannian spin manifold without boundary,
let $S\to M$ be its spinor bundle, let $\cH=L^2(M,S)$, and let $\cD$ be
the corresponding Dirac operator. We introduce a finite family of bounded
operators $\{L_k\}\subset \mathcal{B}(\cH)$ and define the dissipation
operator
\begin{equation}
\Sig:=\frac{1}{2}\sum_k L_k^\dagger L_k.
\end{equation}
In the present framework the Lindblad background enters only through this
bounded datum and its commutator with $\cD$.

\begin{Definition}
\label{def:lindbladatum}
Let $(\cA, \cH, \cD)$ be a spectral triple on a compact Riemannian spin
manifold $(M^d, g)$ without boundary. A \emph{Lindblad datum} for this
triple is a finite collection $\{L_k\} \subset \mathcal{B}(\cH)$ of
bounded operators satisfying:
\begin{itemize}[leftmargin=2.2em]
  \item[(H1)] $\Sig := \tfrac{1}{2}\sum_k L_k^\dagger L_k \in
              \mathcal{B}(\cH)$ is self-adjoint;
  \item[(H2)] $[\cD, \Sig]$, initially defined on $\Dom(\cD)$, extends
              to a bounded operator on $\cH$.
\end{itemize}
The associated \emph{deformed Dirac operator} is
\begin{equation}
\Dg:=\cD-i\gamma\Sig, \qquad \gamma\ge 0.
\end{equation}
Since $\Sig$ is bounded, $\Dg$ is closed on $\Dom(\cD)$ and has the same
domain as $\cD$.
\end{Definition}

\begin{Proposition}
\label{prop:Qgamma}
Under \textup{(H1)--(H2)}, the operator $\Dg := \cD - i\gamma\Sig$ is
closed on $\Dom(\cD)$ with adjoint $\Dg^* = \cD + i\gamma\Sig$, and the
positive self-adjoint operator $\Qg := \Dg^*\Dg$ is given on
$\Dom(\cD^2)$ by
\begin{equation}
  \Qg = \cD^2 + \gamma^2\Sig^2 - i\gamma[\cD,\Sig].
  \label{eq:Qgformula}
\end{equation}
\end{Proposition}

\begin{proof}
Because $\Sig\in \mathcal{B}(\cH)$ is bounded and self-adjoint, standard
perturbation theory gives
\begin{equation}
\Dg^*=\cD+i\gamma\Sig
\qquad \text{on }\Dom(\cD).
\end{equation}
Expanding on $\Dom(\cD^2)$,
\begin{align}
\Qg
&=(\cD+i\gamma\Sig)(\cD-i\gamma\Sig)\notag\\
&=\cD^2-i\gamma\cD\Sig+i\gamma\Sig\cD+\gamma^2\Sig^2\notag\\
&=\cD^2+\gamma^2\Sig^2+i\gamma(\Sig\cD-\cD\Sig)\notag\\
&=\cD^2+\gamma^2\Sig^2-i\gamma[\cD,\Sig].
\end{align}
Since $[\cD,\Sig]^*=-[\cD,\Sig]$ by self-adjointness of $\cD$ and
$\Sig$, the perturbation $-i\gamma[\cD,\Sig]+\gamma^2\Sig^2$ is bounded
and self-adjoint. Therefore $\Qg$ is self-adjoint on $\Dom(\cD^2)$ by the
Kato--Rellich theorem~\cite{Vassilevich2003}. Positivity follows from
\begin{equation}
\langle \psi,\Qg\psi\rangle=\langle \Dg\psi,\Dg\psi\rangle=
\|\Dg\psi\|^2\ge 0.\qedhere
\end{equation}
\end{proof}

\subsection{The triviality of the naive Lindblad heat trace}

Let $\cL$ be a Lindblad generator acting on density matrices. If $\rho_0$
is normalised, then
\begin{equation}
\Tr(e^{t \cL}\rho_0)=1
\end{equation}
for all $t\ge 0$ by trace
preservation~\cite{Lindblad1976,GoriniKossakowskiSudarshan1976}. This
quantity is therefore spectrally trivial and should not be confused with
the heat trace of a Laplace-type operator. The natural spectral-geometric
object in the present framework is instead
\begin{equation}
K_\gamma(\sigma):=\Tr\!\left(e^{-\sigma\Qg}\right).
\end{equation}
This distinction is essential: the present paper studies the heat trace
of $\Qg$, not the physical Lindblad semigroup itself.

\section{Heat-kernel asymptotics for $\mathcal{Q}_\gamma$}
\label{sec:heatkernel}

From this point onward, when discussing local heat-kernel asymptotics
and Seeley--DeWitt coefficients, we restrict to Lindblad data for which
$\Sig$ is a \emph{smooth self-adjoint bundle endomorphism} of the spinor
bundle $S\to M$. Under this assumption $[\cD,\Sig]$ is a zeroth-order
differential operator, so both $\gamma^2\Sig^2$ and $-i\gamma[\cD,\Sig]$
are smooth zeroth-order endomorphism-valued perturbations. The scalar
deformation $L=\sqrt{\gamma}\,f$ with $f\in C^\infty(M)$ is a
prototypical instance of this setting, since $\Sig=\frac{\gamma}{2}f^2$
is then multiplication by a smooth real function.

The principal symbol of $\Qg$ is the same as that of $\cD^2$, namely
$|\xi|^2\mathbf{1}_S$. Under the smooth endomorphism assumption the two
additional terms, $\gamma^2\Sig^2$ and $-i\gamma[\cD,\Sig]$, are
bounded smooth zeroth-order endomorphism-valued perturbations. Hence
$\Qg$ is a Laplace-type operator in the standard
heat-kernel sense~\cite{Vassilevich2003,Gilkey1995}. It follows that
\begin{equation}
K_\gamma(\sigma)\sim (4\pi\sigma)^{-d/2}\sum_{m\ge 0}\sigma^m A_{2m}(\gamma),
\qquad \sigma\to 0^+,
\end{equation}
with dissipation-modified even coefficients $A_{2m}(\gamma)$. The
deformation changes only the lower-order endomorphism term in the
Laplace-type decomposition, so the standard Seeley--DeWitt machinery
applies without modification, and no odd coefficients appear on a closed
manifold.

\subsection{Scalar power counting}
\label{sec:powercounting}

For the scalar deformation
\begin{equation}
L=\sqrt{\gamma}\,f, \qquad f=f^*\in C^\infty(M),
\end{equation}
one has
\begin{equation}
\Sig_\gamma:=\frac{\gamma}{2}\, f^2.
\end{equation}
Note that in this specialization $\Sig_\gamma$ itself depends on $\gamma$, so
$\Dg$ and $\Qg$ are to be read as $\mathcal{D}[\Sigma_\gamma]$ and
$\mathcal{Q}[\Sigma_\gamma]$ respectively. The perturbative hierarchy
below reflects this combined $\gamma$-scaling. We choose this convention
so that the commutator sector enters $\Qg$ at order $\gamma^2$ and the
multiplicative sector at order $\gamma^4$, yielding a clean two-scale
perturbative structure. Therefore
\begin{equation}
-i\gamma[\cD,\Sig_\gamma]=-\frac{i\gamma^2}{2}[\cD,f^2],
\qquad
\gamma^2\Sig_\gamma^2=\frac{\gamma^4}{4}f^4.
\end{equation}
It is thus natural to write
\begin{equation}
\Qg=\cD^2+\gamma^2 W_1+\gamma^4 W_2,
\qquad
W_1:=-\frac{i}{2}[\cD,f^2],
\qquad
W_2:=\frac{1}{4} f^4.
\label{eq:Qg_scalar_split}
\end{equation}
The deformation generated by the commutator term appears already at order
$\gamma^2$, whereas the purely multiplicative contribution enters only at
order $\gamma^4$.

\section{Scalar Lindblad deformation and the vanishing proposition}
\label{sec:vanishing}

For smooth functions on a spin manifold one has the standard identity
\begin{equation}
[\cD,f]=c(df),
\end{equation}
where $c(df)$ denotes Clifford multiplication by the differential
$df$~\cite{BerlineGetzlerVergne1992,LawsonMichelsohn1989}.
It follows that
\begin{equation}
[\cD,f^2]=2f\,c(df),
\end{equation}
and hence
\begin{equation}
W_1=-if\,c(df),
\qquad
W_2=\frac{1}{4} f^4.
\end{equation}
The Duhamel expansion of the deformed heat trace then begins as
\begin{equation}
K_\gamma(\sigma)=K_0(\sigma)+\Delta_1(\sigma)+O(\gamma^4),
\qquad
K_0(\sigma):=\Tr(e^{-\sigma \cD^2}),
\end{equation}
with first correction
\begin{equation}
\Delta_1(\sigma):=-\gamma^2\sigma\,\Tr(W_1e^{-\sigma\cD^2}).
\end{equation}

\begin{Proposition}[Vanishing of the first Duhamel correction]
\label{prop:vanishing}
Let $(\cA, \cH, \cD)$ be a spectral triple equipped with Lindblad datum
$\{L\}$ in the sense of Definition~\ref{def:lindbladatum}, with
$L = \sqrt{\gamma}\,f$, $f = f^* \in C^\infty(M)$. Then
\begin{equation}
  \Delta_1(\sigma) := -\gamma^2\sigma\,
  \Tr\!\left(W_1\,e^{-\sigma\cD^2}\right) = 0
  \qquad \text{for all } \sigma > 0.
\end{equation}
Equivalently, the first nonzero perturbative correction to
$K_\gamma(\sigma)$ enters at order $\gamma^4$.
\end{Proposition}

\begin{proof}
Since $W_1 = -\tfrac{i}{2}[\cD, f^2]$ and $[\cD, e^{-\sigma\cD^2}] = 0$,
\begin{equation}
  \Tr\bigl([\cD, f^2]\,e^{-\sigma\cD^2}\bigr)
  = \Tr\bigl([\cD,\, f^2 e^{-\sigma\cD^2}]\bigr) = 0.
\end{equation}
The first equality uses $[\cD,e^{-\sigma\cD^2}]=0$, since
$e^{-\sigma\cD^2}$ is a function of $\cD^2$. To justify the second
equality, note that $f^2$ is bounded and $e^{-\sigma\cD^2}$ is trace
class for every $\sigma>0$ on a compact manifold, hence
$f^2e^{-\sigma\cD^2}\in \mathcal{L}^1(\cH)$. Moreover,
\begin{equation}
\cD f^2 e^{-\sigma\cD^2}=
\bigl(\cD\,e^{-\sigma\cD^2/2}\bigr)\!\bigl(e^{-\sigma\cD^2/2}f^2\bigr),
\end{equation}
and both factors are Hilbert--Schmidt (operators whose Hilbert--Schmidt
norm $\|T\|_{\mathrm{HS}}^2=\Tr(T^*T)<\infty$), which follows from the
spectral resolution of $\cD$ and the finiteness of
$\sum_n \lambda_n^2 e^{-\sigma \lambda_n^2}$. Hence
$\cD f^2 e^{-\sigma\cD^2}$ is trace class, and similarly
$f^2e^{-\sigma\cD^2}\cD$ is trace class. Therefore the trace of the
commutator vanishes, giving $\Delta_1(\sigma)=0$.
\end{proof}

\begin{Remark}
The cyclicity argument above provides the clean operator-theoretic proof.
A complementary geometric viewpoint is consistent with the same conclusion:
for any smooth real $f$, the commutator $[\cD,f]=c(df)$ is Clifford
multiplication by $df$, and the fibrewise spinor trace $\tr_S(c(\xi))=0$
for every $\xi\in T_x^*M$, since $c(\xi)$ is an odd Clifford element
whose eigenvalues $\pm i|\xi|$ cancel in the trace. This
\emph{Clifford-tracelessness} is consistent with the exact vanishing
of $\Delta_1(\sigma)$.
\end{Remark}

\begin{Corollary}
\label{cor:rigidity}
In the scalar case, $K_\gamma(\sigma) = K_0(\sigma) + O(\gamma^4)$
at fixed $\sigma > 0$. The spectral geometry is perturbatively rigid
at order~$\gamma^2$.
\end{Corollary}

\begin{Corollary}[Coefficient-level rigidity in the scalar case]
\label{cor:SDW}
Let $\Qg=\cD^2+\gamma^2W_1+\gamma^4W_2$ be the scalar Lindblad
deformation, with heat-trace expansion
\begin{equation}
K_\gamma(\sigma)\sim (4\pi\sigma)^{-d/2}\sum_{m\ge 0}\sigma^m A_{2m}(\Qg),
\qquad \sigma\to 0^+.
\label{eq:SDW_expansion}
\end{equation}
Then
\begin{equation}
A_0(\Qg)=A_0(\cD^2),
\label{eq:A0_exact}
\end{equation}
and
\begin{equation}
A_2(\Qg)=A_2(\cD^2)+O(\gamma^4).
\label{eq:A2_order4}
\end{equation}
In particular, there is no order-$\gamma^2$ correction to the first
nontrivial Seeley--DeWitt coefficient in the scalar deformation.
\end{Corollary}

\begin{proof}
The equality~\eqref{eq:A0_exact} follows because the principal symbol of
$\Qg$ equals that of $\cD^2$, namely $|\xi|^2\mathbf{1}_S$, so the
leading heat coefficient, which depends only on this symbol (and hence
on the volume and spinor-bundle rank), is unchanged by the bounded
lower-order perturbation.

For~\eqref{eq:A2_order4}, note that the entire order-$\gamma^2$
contribution to the heat trace vanishes identically by
Proposition~\ref{prop:vanishing}, so its small-$\sigma$ asymptotic
expansion vanishes coefficientwise. In particular, $A_2(\Qg)$ receives no
$O(\gamma^2)$ correction. Combined with Corollary~\ref{cor:rigidity},
which gives $K_\gamma(\sigma)=K_0(\sigma)+O(\gamma^4)$, it follows that
$A_2(\Qg)=A_2(\cD^2)+O(\gamma^4)$.
\end{proof}

\section{Structure of the first nonzero correction}
\label{sec:gamma4}

Expanding the Duhamel series for \eqref{eq:Qg_scalar_split} to order
$\gamma^4$ yields an exact decomposition of the first nonzero correction
into two distinct sectors. We first define
\begin{equation}
\Delta_{2,a}(\sigma):=-\gamma^4\sigma\,\Tr(W_2\,e^{-\sigma\cD^2}),
\label{eq:Delta2a_def}
\end{equation}
coming from the direct first-order insertion of $W_2$, and
\begin{equation}
\Delta_{2,b}(\sigma):=\gamma^4\!\int_0^\sigma\!\!\int_0^s
\Tr\!\left(e^{-(\sigma-s)\cD^2}W_1\,e^{-(s-r)\cD^2}W_1\,e^{-r\cD^2}
\right)\,dr\,ds,
\label{eq:Delta2b_def}
\end{equation}
coming from the quadratic Duhamel term in $W_1$.

\begin{Proposition}[Exact structure of the order-$\gamma^4$ correction]
\label{prop:gamma4structure}
Let $\Qg=\cD^2+\gamma^2W_1+\gamma^4W_2$ as in \eqref{eq:Qg_scalar_split}.
Then the heat trace admits the expansion
\begin{equation}
K_\gamma(\sigma)=K_0(\sigma)+\Delta_{2,a}(\sigma)+\Delta_{2,b}(\sigma)
+O(\gamma^6).
\label{eq:Kgamma_gamma4}
\end{equation}
\end{Proposition}

\begin{proof}
This follows from the Duhamel expansion recorded in
Appendix~\ref{app:duhamel}, together with Proposition~\ref{prop:vanishing},
which eliminates the order-$\gamma^2$ term.
\end{proof}

Thus the first nonzero perturbative correction is completely exhausted by
the direct $W_2$ insertion and the quadratic $W_1\otimes W_1$ term.

The decomposition~\eqref{eq:Kgamma_gamma4} is exact at the Duhamel level;
only the local small-$\sigma$ evaluations below are asymptotic. For the
two-dimensional rank-$2$ spinor setting relevant to the later $S^2$
specialization, the leading local heat-kernel asymptotic gives
\begin{equation}
\Tr(V\,e^{-\sigma \cD^2})\sim \frac{1}{4\pi\sigma}\int_M
\tr_S(V(x))\,d\mathrm{vol}_g + O(\sigma^0),
\qquad \sigma\to 0^+.
\label{eq:localasymp}
\end{equation}
Applying \eqref{eq:localasymp} to $W_2=\tfrac{1}{4}f^4\mathbf{1}_S$,
and using $\tr_S(W_2)=\tfrac{1}{2}f^4$ for a rank-$2$ spinor bundle,
gives the general formula
\begin{equation}
\Delta_{2,a}(\sigma)\sim
-\frac{\gamma^4}{8\pi}\int_M f^4\,d\mathrm{vol}_g
+O(\gamma^4\sigma).
\label{eq:W2leading}
\end{equation}
Similarly, the quadratic $W_1\otimes W_1$ sector contributes a leading
term of order $\gamma^4\sigma$ in two dimensions. In the absence of a
globally fixed Clifford normalisation and exact spinor harmonic evaluation,
we leave its signed local coefficient in the form
\begin{equation}
\Delta_{2,b}(\sigma)\sim C_{W_1W_1}\,\gamma^4\sigma
+O(\gamma^4\sigma^2),
\label{eq:W1W1leading}
\end{equation}
where $C_{W_1W_1}$ is determined by the detailed local Clifford
representation and the exact spectral data. These are leading asymptotic
statements, not exact identities. The exact order-$\gamma^4$ coefficient
requires either a full spinor spherical harmonic computation or a
controlled pseudodifferential expansion, and is deferred.

\begin{Remark}[Leading local structure of the two $O(\gamma^4)$ sectors]
\label{rem:localinvariants}
At the level of leading local heat-kernel structure, the two
$O(\gamma^4)$ sectors are naturally associated with the following scalar
integrals. The direct $W_2$ sector involves
\begin{equation}
\int_M f^4\,d\mathrm{vol}_g,
\end{equation}
since $W_2 = \tfrac{1}{4}f^4$ and the local asymptotic~\eqref{eq:localasymp}
gives a contribution proportional to $\int_M \tr_S(W_2)\,d\mathrm{vol}_g$.
The quadratic $W_1\otimes W_1$ sector, at the level of the leading
diagonal approximation, involves
\begin{equation}
\int_M f^2|\nabla f|^2\,d\mathrm{vol}_g,
\end{equation}
since $W_1 = -if\,c(df)$ gives
\begin{equation}
W_1^2=(-i)^2 f^2 c(df)^2=-f^2 c(df)^2=f^2|df|_g^2\mathbf{1}_S,
\end{equation}
using $c(\alpha)^2=-|\alpha|_g^2\mathbf{1}_S$, and therefore
$\tr_S(W_1^2)=\mathrm{rk}(S)\cdot f^2|df|_g^2\ge 0$.
A full determination of whether additional local invariants
(for example curvature couplings) contribute at this order would require
a controlled pseudodifferential coefficient calculation and is not
attempted here.
\end{Remark}

\section{The $S^2$ toy model}
\label{sec:S2}

We now specialise to the round two-sphere of unit radius. The scalar
curvature is $R=2$, the volume is $\mathrm{vol}(S^2)=4\pi$, and the
spinor bundle has rank $2$~\cite{LawsonMichelsohn1989}. The Dirac spectrum
on the round sphere is well known~\cite{CamporesiHiguchi1996}: the
eigenvalues of $\cD$ are
\begin{equation}
\lambda_n^\pm=\pm(n+1),
\qquad n=0,1,2,\dots,
\end{equation}
with multiplicity $2(n+1)$ for each sign. Therefore $\cD^2$ has
eigenvalues
\begin{equation}
\mu_n=(n+1)^2
\end{equation}
with multiplicity $4(n+1)$, and
\begin{equation}
K_0(\sigma)=\sum_{n=0}^\infty 4(n+1)\,e^{-\sigma(n+1)^2}.
\end{equation}
Its small-$\sigma$ asymptotic expansion begins as
\begin{equation}
K_0(\sigma)\sim \frac{2}{\sigma}-\frac{1}{3}+O(\sigma),
\qquad \sigma\to 0^+.
\label{eq:K0s2}
\end{equation}

We choose the scalar mode
\begin{equation}
f(\theta,\phi)=\cos\theta,
\end{equation}
which is proportional to the spherical harmonic $Y_{10}$. This is the
simplest nontrivial scalar deformation and should not be confused with the
unit normal vector field of the embedding $S^2\subset \mathbb{R}^3$.

For this choice,
\begin{equation}
\Sig_\gamma=\frac{\gamma}{2}\cos^2\!\theta\;\mathbf{1}_2.
\end{equation}
Using the standard local orthonormal coframe
\begin{equation}
e^1=d\theta,\qquad e^2=\sin\theta\,d\phi,
\end{equation}
one has
\begin{equation}
[\cD,\Sig]
=c\!\left(d\!\left(\tfrac{\gamma}{2}\cos^2\theta\right)\right)
=-\gamma\sin\theta\cos\theta\,c(e^1)
=-\gamma\sin\theta\cos\theta\,c(d\theta),
\end{equation}
and the deformed operator therefore takes the local form
\begin{equation}
\Qg
=\cD^2+i\gamma^2\sin\theta\cos\theta\,c(d\theta)
+\frac{\gamma^4}{4}\cos^4\!\theta\,\mathbf{1}_2.
\label{eq:QgS2}
\end{equation}
This expression is explicit at the operator level and avoids committing to
a particular Pauli-matrix realisation of the local Clifford action. It is
the appropriate starting point for a full $S^2$ spectral computation.
Here $\mathbf{1}_2$ denotes the identity on the rank-$2$ spinor fibre,
i.e.\ the local representative of $\mathbf{1}_S$ on $S^2$.

At the level of leading local small-$\sigma$ asymptotics, the two
$O(\gamma^4)$ sectors take the form
\begin{equation}
\Delta_{2,a}(\sigma)\sim -\frac{\gamma^4}{10}+O(\gamma^4\sigma),
\end{equation}
where the constant $-\gamma^4/10$ is the specialization of the general
formula~\eqref{eq:W2leading} to $f=\cos\theta$ on $S^2$, using
$\int_{S^2}\cos^4\theta\,d\Omega = \tfrac{4\pi}{5}$. For the quadratic
sector one similarly has
\begin{equation}
\Delta_{2,b}(\sigma)\sim C_{W_1W_1}\,\gamma^4\sigma+O(\gamma^4\sigma^2),
\end{equation}
where the coefficient $C_{W_1W_1}$ depends on the detailed local Clifford
normalisation and the exact spinor spectral data. These are local
asymptotic statements rather than exact spectral identities. A fully
rigorous determination of the exact $O(\gamma^4)$ coefficient requires
either an explicit spinor spherical harmonic computation or a controlled
pseudodifferential expansion, and is deferred.

\section{Effective scale-dependent spectral dimension}
\label{sec:spectraldim}

Motivated by the use of spectral dimension as a scale-dependent probe in
quantum-gravity-inspired
settings~\cite{AmbjornJurkiewiczLoll2005,AlkoferSaueressigZanusso2015},
we define the effective spectral dimension associated with the deformed
heat trace by
\begin{equation}
d_{s,\mathrm{eff}}(\sigma,\gamma):=-2\,\frac{\partial}{\partial\log\sigma}
\log K_\gamma(\sigma).
\label{eq:dseff_def}
\end{equation}
This is a diagnostic quantity attached to the heat trace of $\Qg$; it is
not to be identified automatically with ultraviolet spectral dimensions
arising in other quantum-gravity frameworks.

Using the perturbative expansion
\begin{equation}
K_\gamma(\sigma)=K_0(\sigma)+\gamma^4\Delta_2(\sigma)+O(\gamma^6),
\qquad
\Delta_2(\sigma):=\frac{\Delta_{2,a}(\sigma)+\Delta_{2,b}(\sigma)}{\gamma^4},
\end{equation}
one finds formally
\begin{equation}
d_{s,\mathrm{eff}}(\sigma,\gamma)=d_{s,\mathrm{eff}}(\sigma,0)
-2\gamma^4\frac{\partial}{\partial \log\sigma}\!\left(
\frac{\Delta_2(\sigma)}{K_0(\sigma)}\right)+O(\gamma^6).
\label{eq:dseffformal}
\end{equation}
In the undeformed $S^2$ model, using~\eqref{eq:K0s2},
\begin{equation}
d_{s,\mathrm{eff}}(\sigma,0)=2+\frac{1}{3}\sigma+O(\sigma^2),
\qquad \sigma\to 0^+.
\label{eq:dseff0}
\end{equation}
Thus the dissipative deformation begins at order $\gamma^4$, mirroring
Corollary~\ref{cor:rigidity}.

It is worth noting that the two contributions to $\Delta_2(\sigma)$ enter
$d_{s,\mathrm{eff}}$ in different ways. In the $S^2$ specialization, the
$W_2$ sector contributes the constant term $-\gamma^4/10$ to the
small-$\sigma$ heat-trace asymptotics, cf.~\eqref{eq:W2leading}. Although
this term is constant at the level of $K_\gamma(\sigma)$, it is not
invisible in $d_{s,\mathrm{eff}}(\sigma,\gamma)$ because $K_0(\sigma)$ is
$\sigma$-dependent. Using $K_0(\sigma)\sim 2/\sigma$ at leading order,
\begin{equation}
-2\gamma^4 \frac{\partial}{\partial \log\sigma}
\!\left(\frac{-1/10}{K_0(\sigma)}\right)
\sim +\frac{\gamma^4\sigma}{10},
\qquad \sigma\to 0^+.
\label{eq:W2dseff}
\end{equation}
Thus in the $S^2$ model the $W_2$ term contributes at order $\gamma^4\sigma$
to the perturbative deformation of $d_{s,\mathrm{eff}}$. The corresponding
$W_1\otimes W_1$ contribution is encoded by the leading local asymptotic
estimate~\eqref{eq:W1W1leading}. Since the exact order-$\gamma^4$
coefficient remains deferred, we do not combine both sectors into a single
closed formula here.

\section{Discussion and outlook}
\label{sec:discussion}

This paper introduced a Lindblad-deformed spectral geometric framework
and established its first analytic consequences. Starting from the
deformed Dirac operator $\Dg = \cD - i\gamma\Sig$, we showed in
Section~\ref{sec:deformation} that the associated positive operator
$\Qg = \Dg^*\Dg = \cD^2 + \gamma^2\Sig^2 - i\gamma[\cD,\Sig]$ is
self-adjoint under the boundedness hypotheses~(H1)--(H2), and in the
smooth endomorphism-valued setting of Section~\ref{sec:heatkernel} is of
Laplace type and supports a standard local heat-kernel expansion.

A key preliminary observation is that the naive
Lindblad trace $\Tr(e^{t\cL}\rho_0) = 1$ is spectrally trivial by
trace preservation, which motivates working with
$K_\gamma(\sigma) = \Tr(e^{-\sigma\Qg})$ as the correct geometric probe.

In Section~\ref{sec:vanishing} we proved that for scalar Lindblad data
$L = \sqrt{\gamma}\,f$ with $f$ real-valued, the first-order Duhamel
correction $\Delta_1(\sigma)$ vanishes identically. The proof rests on
cyclicity of the trace together with the trace-class properties of
$\cD f^2 e^{-\sigma\cD^2}$, which follow from the Hilbert--Schmidt
factorisation on a compact manifold. A complementary geometric viewpoint
--- the Clifford-tracelessness of $c(df)$ --- is consistent with the
exact vanishing. The upshot, recorded as Corollary~\ref{cor:rigidity}, is
that the spectral geometry is perturbatively rigid at order $\gamma^2$,
and the first nonzero dissipative correction to $K_\gamma(\sigma)$ enters
at order $\gamma^4$. Corollary~\ref{cor:SDW} upgrades this to a
coefficient-level statement: $A_0(\Qg)=A_0(\cD^2)$
exactly, and $A_2(\Qg)=A_2(\cD^2)+O(\gamma^4)$,
so there is no order-$\gamma^2$ correction to the first nontrivial
Seeley--DeWitt coefficient.

In Sections~\ref{sec:gamma4} and~\ref{sec:S2} we identified the two
distinct $O(\gamma^4)$ contributions: a direct $W_2$ insertion
contributing, in the round $S^2$ model with $f=\cos\theta$, the constant
$-\gamma^4/10$ to the leading local heat-trace
asymptotics, and a quadratic $W_1\otimes W_1$ term contributing at
leading order linearly in $\sigma$ in two dimensions.
Remark~\ref{rem:localinvariants} identifies, at the level of leading
local structure, the scalar integrals $\int_M f^4\,d\mathrm{vol}_g$
and $\int_M f^2|\nabla f|^2\,d\mathrm{vol}_g$ as the natural geometric
quantities arising from these two sectors; a full determination of all
contributing local invariants is deferred. These results are asymptotic
statements based on the local heat-kernel expansion and are not exact
spectral identities. In Section~\ref{sec:spectraldim} we
defined the effective scale-dependent spectral dimension and noted that
the $W_2$ constant term, while invisible at the level of $K_\gamma$
itself, contributes an $O(\gamma^4\sigma)$ correction to
$d_{s,\mathrm{eff}}$ through the $\sigma$-dependence of $K_0(\sigma)$.

Several directions remain open. The most immediate is the exact
determination of the order-$\gamma^4$ coefficient in the $S^2$ model,
which requires either a full spinor spherical harmonic computation or a
controlled pseudodifferential expansion; the leading local asymptotic
results obtained here provide the correct structural decomposition but
not the exact numerical coefficient. A separate question concerns the
precise relationship between the physical Lindblad semigroup acting on
density matrices and the spectral-geometric object $\Qg$ studied here;
constructing such a bridge would clarify the physical interpretation of
$K_\gamma(\sigma)$ and its effective dimension. More broadly, one could
investigate whether the Lindblad deformation has topological or
index-theoretic consequences, and whether a complete open-system
reformulation of the spectral triple axioms is tractable. It would also
be natural to explore how the framework extends beyond the scalar
deformation case, for instance to matrix-valued or geometry-dependent
jump operators, and whether the perturbative scale-dependent deformation of
$d_{s,\mathrm{eff}}$ identified here persists or takes a different form in
higher dimensions or on manifolds of nontrivial topology.

\appendix

\section{Derivation of $\mathcal{Q}_\gamma$:
  sign verification}
\label{app:signs}

For completeness we record the sign computation explicitly.

In a local orthonormal coframe $e^a$, we take the Riemannian Clifford
action to satisfy
\begin{equation}
c(e^a)c(e^b)+c(e^b)c(e^a)=-2\delta^{ab}\mathbf{1},
\end{equation}
so that a concrete Pauli-matrix realisation may be written as
$c(e^a)=i\sigma^a$. When convenient, this identifies the local Clifford
action with the standard Pauli-matrix realisation. Since
$\Dg=\cD-i\gamma\Sig$ and $\Dg^*=\cD+i\gamma\Sig$,
\begin{align}
\Qg&=(\cD+i\gamma\Sig)(\cD-i\gamma\Sig)\notag\\
&=\cD^2-i\gamma\cD\Sig+i\gamma\Sig\cD+\gamma^2\Sig^2\notag\\
&=\cD^2+\gamma^2\Sig^2+i\gamma(\Sig\cD-\cD\Sig)\notag\\
&=\cD^2+\gamma^2\Sig^2-i\gamma[\cD,\Sig].
\end{align}
Moreover,
\begin{equation}
([\cD,\Sig])^*= (\cD\Sig-\Sig\cD)^*=\Sig\cD-\cD\Sig=-[\cD,\Sig],
\end{equation}
so $-i[\cD,\Sig]$ is self-adjoint. In the round $S^2$ scalar model,
\begin{equation}
\Sig_\gamma=\frac{\gamma}{2}\cos^2\!\theta\;\mathbf{1}_2,
\qquad
[\cD,\Sig_\gamma]=-\gamma\sin\theta\cos\theta\,c(d\theta).
\end{equation}
Substituting this into \eqref{eq:Qgformula} yields the local operator form
displayed in \eqref{eq:QgS2}. If one chooses an explicit Pauli-matrix
realisation of the Clifford action, the commutator term is represented by
the corresponding local matrix-valued Clifford generator.

\section{Duhamel expansion to order $\gamma^4$}
\label{app:duhamel}

Write $\Qg=\cD^2+\gamma^2W_1+\gamma^4W_2$. The Duhamel formula gives
\begin{align}
e^{-\sigma\Qg}
&=e^{-\sigma\cD^2}
-\gamma^2\!\int_0^\sigma e^{-(\sigma-s)\cD^2}W_1e^{-s\cD^2}\,ds
-\gamma^4\!\int_0^\sigma e^{-(\sigma-s)\cD^2}W_2e^{-s\cD^2}\,ds
\notag\\
&\quad +\gamma^4\!\int_0^\sigma\!\!\int_0^s
e^{-(\sigma-s)\cD^2}W_1e^{-(s-r)\cD^2}W_1e^{-r\cD^2}\,dr\,ds
+O(\gamma^6).
\label{eq:duhamel_full}
\end{align}
Taking traces and using cyclicity yields
\begin{align}
K_\gamma(\sigma)
&=K_0(\sigma)
-\gamma^2\sigma\,\Tr(W_1e^{-\sigma\cD^2})
-\gamma^4\sigma\,\Tr(W_2e^{-\sigma\cD^2})\notag\\
&\quad +\gamma^4\!\int_0^\sigma\!\!\int_0^s
\Tr\!\left(e^{-(\sigma-s)\cD^2}W_1e^{-(s-r)\cD^2}W_1e^{-r\cD^2}\right)
dr\,ds+O(\gamma^6).
\end{align}
The vanishing of the $\gamma^2$ term follows from
Proposition~\ref{prop:vanishing}. For the $W_2$ term, we specialize to
the $S^2$ model with $f=\cos\theta$. Since
$W_2=\tfrac{1}{4}\cos^4\theta\,\mathbf{1}_2$ and
$\tr_S(W_2)=\tfrac{1}{2}\cos^4\theta$ for the rank-$2$ spinor bundle,
applying~\eqref{eq:localasymp} gives
\begin{equation}
\Tr(W_2\,e^{-\sigma\cD^2})\sim
\frac{1}{4\pi\sigma}\cdot\frac{1}{2}\int_{S^2}\cos^4\theta\,d\Omega
=\frac{1}{4\pi\sigma}\cdot\frac{1}{2}\cdot\frac{4\pi}{5}
=\frac{1}{10\sigma},
\end{equation}
giving $\Delta_{2,a}(\sigma)\sim -\gamma^4/10$, which is the
specialization of the general formula~\eqref{eq:W2leading} to this
model. The prefactor $-\gamma^4\sigma$ in $\Delta_{2,a}(\sigma)$ produces
a constant leading contribution in two dimensions.

For the $W_1\otimes W_1$ term, a leading local approximation is obtained
by replacing the two heat factors by a single diagonal heat kernel at
scale $\sigma$, yielding a contribution proportional to
$\Tr(W_1^2e^{-\sigma\cD^2})$. In the scalar model, $W_1=-if\,c(df)$, so
the relevant local endomorphism is controlled by $f^2c(df)^2$. The
precise signed coefficient depends on the chosen normalisation of the
local Riemannian Clifford representation and on the exact spinor harmonic
data. Since the present paper does not carry out that exact
normalisation-dependent computation globally, we record only the
structural conclusion that in two dimensions the quadratic sector
contributes at order $\gamma^4\sigma$:
\begin{equation}
\Delta_{2,b}(\sigma)\sim C_{W_1W_1}\,\gamma^4\sigma
+O(\gamma^4\sigma^2),
\end{equation}
for a local coefficient $C_{W_1W_1}$ to be fixed by a full spectral
calculation. The formulas in Section~\ref{sec:gamma4} are exact at the
Duhamel level; only the subsequent small-$\sigma$ evaluations are
asymptotic.


\bibliographystyle{plain}
\bibliography{bibliography}

\end{document}